\pdfoutput=1
\documentclass[a4paper,USenglish,cleveref]{lipics-v2021}

\usepackage{algorithm,algpseudocode}

\usepackage{url}
\usepackage{amsmath}
\usepackage{amsthm}
\usepackage{array}
\usepackage{subcaption}
\usepackage{graphicx}

\usepackage{mathpartir}

\usepackage{extarrows}

\usepackage{color}

\newcommand{\states}{\mathbb{Q}}
\newcommand{\msgs}{\mathbb{M}\mathit{sgs}}

\newcommand{\tup}[1]{\langle{#1}\rangle}

\newcommand{\acts}{\mathbb{A}}
\newcommand{\call}[1]{\mathit{call}\ {#1}}
\newcommand{\ret}[1]{\mathit{ret}\ {#1}}

\newcommand{\mpimp}[1]{I_{mp}({#1})}
\newcommand{\smimp}[1]{I_{sm}({#1})}

\nolinenumbers

\title{Impossibility of Strongly-Linearizable Message-Passing Objects
                     via Simulation by Single-Writer Registers}
\titlerunning{Impossibility of Strongly-Linearizable Message-Passing Objects}

\author{Hagit Attiya}{Computer Science Department, Technion}{}{0000-0002-8017-6457}{Partially supported by the Israel Science Foundation (grant number 380/18).} 
\author{Constantin Enea}{Université de Paris, IRIF, CNRS}{}{0000-0003-2727-8865}{Supported in part by the European Research Council (ERC) under the European Union's Horizon 2020 research and innovation program (grant agreement No 678177).}  
\author{Jennifer L.~Welch}{Texas A\&M University}{}{0000-0003-2725-9875}{Supported in part by the U.S.\ National Science Foundation under grant number 1816922.}  

\authorrunning{Attiya, Enea and Welch}

\Copyright{Hagit Attiya, Constantin Enea, and Jennifer L.~Welch}

\acknowledgements{We thank the anonymous referees for helpful
comments that improved the presentation of the paper.}



\EventEditors{Seth Gilbert}
\EventNoEds{1}
\EventLongTitle{35th International Symposium on Distributed Computing (DISC 2021)}
\EventShortTitle{DISC 2021}
\EventAcronym{DISC}
\EventYear{2021}
\EventDate{October 4--8, 2021}
\EventLocation{Freiburg, Germany (Virtual Conference)}
\EventLogo{}
\SeriesVolume{209}
\ArticleNo{20}

\hideLIPIcs

\begin{document}

\maketitle

\begin{abstract}
A key way to construct complex distributed systems is through modular
composition of linearizable concurrent objects.
A prominent example is shared registers, which have crash-tolerant
implementations on top of message-passing systems, allowing the advantages
of shared memory to carry over to message-passing.
Yet linearizable registers do not always behave properly when used inside
randomized programs.
A strengthening of linearizability, called strong linearizability,
has been shown to preserve probabilistic behavior, as well as other
``hypersafety'' properties.
In order to exploit composition and abstraction in message-passing systems,
it is crucial to know whether there exist strongly-linearizable
implementations of registers in message-passing.
This paper answers the question in the negative: \emph{there are no
strongly-linearizable fault-tolerant message-passing implementations
of multi-writer registers, max-registers, snapshots or counters}.
This result is proved by reduction from the corresponding result by
Helmi et al.
The reduction is a novel extension of the BG simulation
that connects shared-memory and message-passing, supports
long-lived objects, and preserves strong linearizability.
The main technical challenge arises from the discrepancy between
the potentially minuscule fraction of failures to be tolerated in
the simulated message-passing algorithm and the large fraction
of failures that can afflict the simulating shared-memory system.
The reduction is general and can be viewed as the inverse of the
ABD simulation of shared memory in message-passing.
\end{abstract}

\section{Introduction}

A key way to construct complex distributed systems is through modular
composition of linearizable concurrent objects~\cite{HerlihyW1990}.
A prominent example is the ABD fault-tolerant message-passing
implementation of shared registers~\cite{AttiyaBD1995}
and its multi-writer variant~\cite{LynchS1997}.
In multi-writer ABD, there is a set of client processes,
which accept invocations of methods on the shared register and provide responses,
and a set of server processes, which replicate the (virtual) state of the register.
When a read method is invoked at a client,
the client queries a majority of the servers to obtain the latest value,
as determined by a timestamp, and then sends the
chosen value back to a majority of the servers before returning.
When a write method is invoked, the client queries a majority of the servers to
obtain the latest timestamp, assigns a larger timestamp to the
current value to be written, and then sends the new value and its timestamp
to a majority of the servers before returning.
Client and server processes run on a set of (physical) nodes;
a node may run any combination of client and server processes.
The algorithm tolerates any distribution of process crashes as long
as less than half of the server processes crash; there is no limit on
the number of client processes that crash.

Variations of ABD have been used to simplify the design of numerous
fault-tolerant message-passing algorithms,
by providing the familiar shared-memory abstraction
(e.g., in \emph{Disk Paxos}~\cite{GafniL2003}).
Yet linearizable registers do not always compose correctly
with randomized programs:
In particular,
\cite{HadzilacosHT2021} demonstrates a randomized program
that terminates with constant probability when used with an
\emph{atomic} register, on which methods execute instantaneously,
but an adversary can in principle prohibit it from terminating when the register
is implemented in a message-passing system, where methods are not instantaneous.
The analogous result is shown in \cite{AttiyaEW2021blunting}
specifically for the situation when ABD is the implementation.

\emph{Strong linearizability}~\cite{GolabHW2011},
a restriction of linearizability,
ensures that properties holding when a concurrent program is executed
in conjunction with an atomic object,
continue to hold when the program is executed
with a strongly-linearizable implementation of the object.
Strong linearizability was shown~\cite{AttiyaE2019} to be
necessary and sufficient for preserving
\emph{hypersafety properties}~\cite{ClarksonS2010},
such as security properties and
probability distributions of reaching particular program states.

These observations highlight the importance of knowing whether
there exists a \emph{strongly-linearizable}
fault-tolerant message-passing implementation of a shared register.
If none exists,
it will be necessary to argue about hypersafety properties
without being able to capitalize on the shared-memory abstraction.

This paper brings bad news, answering this question in the negative:
{\em There are no strongly-linearizable fault-tolerant message-passing
implementations of several highly useful objects,
including multi-writer registers, max-registers, snapshots and counters.}

One might be tempted to simply conclude this result from the
impossibility result of Helmi et al.~\cite{HelmiHW12}
showing that there is no strongly-linearizable nonblocking
implementation of a multi-writer register from single-writer registers.
However, reproducing the proof in~\cite{HelmiHW12} for
the message-passing model is not simple,
as it is rather complicated and tailored to the shared-memory model.
In particular, parts of the proof require progress
when a process executes solo, which cannot be easily imitated when
the number of failures is much smaller than the number of processes.

Another approach is to reduce to the impossibility result in the
shared-memory model.
A simple reduction simulates message transfer between each
pair of message-passing nodes using dedicated shared registers.
This simulation uses the same number of shared-memory processes as
message-passing nodes and preserves the number of failures tolerated.
However, message-passing register implementations require the total
number of nodes to be at least twice the number of failures
tolerated~\cite{AttiyaBD1995},
while the proof of Helmi et al.\ critically depends on the fact that
all processes, except perhaps one, may stop taking steps.
It is not obvious how to simulate the workings of many message-passing
nodes, only a small fraction of which may fail, using the same number of
shared-memory processes, almost all of which may fail.

We take a different path to prove this result by reduction,
extending the BG simulation~\cite{BorowskyG1993}
in three nontrivial ways.
First, our reduction works across communication models,
and bridges the gap between shared-memory and message-passing systems.
Second, it supports long-lived objects,
on which each process can invoke any number of methods instead
of just one.
And most importantly, it preserves strong linearizability.

In more detail, we consider a hypothetical strongly-linearizable
message-passing algorithm
that implements a long-lived object.
Following contemporary expositions of such algorithms
(e.g.~\cite{HadjistasiNS2017,HuangHW2020}),
we assume that the algorithm is organized into
a set of $m$ client processes, any number of which may crash,
and $n$ server processes, up to $m-1$ of which may crash,
running on a set of nodes.
We obtain a nonblocking shared-memory implementation
of the same object for $m$ processes, $m-1$ of which may fail,
using \emph{single-writer} registers.

Our implementation admits
a \emph{forward simulation} to the message-passing implementation.
The forward simulation is a relation between states of the two
implementations.
Using the forward simulation, we can construct an execution of
the message-passing implementation from any execution of the shared-memory
implementation, starting from the initial state and moving forward step
by step, such that the two executions have the same sequence of method
invocations and responses.

Since the hypothetical message-passing algorithm is strongly linearizable,
a result from~\cite{AttiyaE2019,Rady} implies that there is a forward
simulation from the message-passing algorithm to the atomic object.
Since forward simulations compose, we obtain a forward simulation
from the shared memory algorithm to the atomic object.
Another result from~\cite{AttiyaE2019,Rady}
shows that a forward simulation implies strong linearizability.
Therefore,
a strongly-linearizable message-passing implementation of a multi-writer
register yields a strongly-linearizable shared-memory implementation
of a multi-writer register using single-writer registers.
Now we can appeal to the impossibility result of~\cite{HelmiHW12}
to conclude that there can be no strongly-linearizable message-passing
implementation of a multi-writer register.
The same argument shows the impossibility of strongly-linearizable
message-passing implementations of max-registers, snapshots, and counters,
which are proved in~\cite{HelmiHW12} to have no strongly-linearizable
implementations using single-writer registers.

We consider the reduction to be interesting in its own right,
because it shows how general message-passing object implementations can
be translated into corresponding shared-memory object implementations.
In this sense, it can be interpreted as an inverse of ABD,
which translates shared-memory object implementations into
message-passing object implementations.
It thus relates the two models,
keeping the same number of failures,
without restricting the total number of processes in the
message-passing model.
We believe it may have additional applications in other contexts.

\section{Objects}\label{sec:objects}

An \emph{object} is defined by a set of method names and an implementation that defines the behavior of each method. Methods can be invoked in parallel at different processes. The executions of an implementation are modeled as sequences of labeled transitions between global states that track the local states of all the participating processes (more precise definitions will be given in Section~\ref{ssec:mp_impls} and Section~\ref{ssec:sm_impls}). Certain transitions of an execution correspond to new invocations of a method or returning from an invocation performed in the past. Such transitions are labeled by call and return actions, respectively.
A \emph{call action} $\call{M(x)}_k$ represents the event of invoking a method $M$ with argument $x$; $k$ is an identifier of this invocation. A \emph{return action} $\ret{y}_k$ represents the event of the invocation $k$ returning value $y$. For simplicity, we assume that each method takes as parameter or returns a single value. We may omit invocation identifiers from call or return actions when they are not important. The set of executions of an object $O$ is denoted by $E(O)$.

\subsection{Object Specifications}

The specification of an object characterizes sequences of call and return actions, called \emph{histories}. The history of an execution $e$, denoted by $\mathit{hist}(e)$, is defined as the projection of $e$ on the call and return actions labeling its transitions. The set of histories of all the executions of an object $O$ is denoted by $H(O)$. Call and return actions $\call{M(x)}_k$ and $\ret{y}_k$ are called \emph{matching} when they contain the same invocation identifier $k$. A call action is called \emph{unmatched} in a history $h$ when $h$ does not contain the matching return. A history $h$ is called \emph{sequential} if every call $\call{M(x)}_k$ is immediately followed by the matching return $\ret{y}_k$. Otherwise, it is called \emph{concurrent}.

\emph{Linearizability}~\cite{HerlihyW1990} expresses the conformance 
of object histories to a given set of sequential histories, 
called a \emph{sequential specification}.
This correctness criterion is based on a relation $\sqsubseteq$ between histories: $h_1\sqsubseteq h_2$ iff there exists a history $h_1'$ obtained from $h_1$ by appending return actions that correspond to some of the unmatched call actions in $h_1$ (completing some pending invocations)
and deleting the remaining
unmatched call actions in $h_1$ (removing some pending invocations), such that $h_2$ is a permutation of $h_1'$ that preserves the order between return and call actions, i.e.,
if a given return action occurs before a given call action in $h_1'$ then the same holds in $h_2$. We say that $h_2$ is a \emph{linearization} of $h_1$.
A history $h_1$ is called \emph{linearizable} w.r.t.~a sequential specification $\mathit{Seq}$ iff there exists a sequential history $h_2\in \mathit{Seq}$ such that $h_1\sqsubseteq h_2$. An object $O$ is linearizable w.r.t.~$\mathit{Seq}$ iff each history $h_1\in H(O)$ is linearizable w.r.t.~$\mathit{Seq}$.

\emph{Strong linearizability}~\cite{GolabHW2011} is a strengthening 
of linearizability which requires that linearizations of an execution 
can be defined in a prefix-preserving manner.
Formally, an object $O$ is \emph{strongly linearizable} 
w.r.t.~$\mathit{Seq}$ iff there exists a function 
$f:E(O)\rightarrow \mathit{Seq}$ such that: 
\begin{enumerate}
\item for any execution $e\in E(O)$, $\mathit{hist}(e)\sqsubseteq f(e)$, and 
\item $f$ is prefix-preserving, i.e., 
for any two executions $e_1,e_2\in E(O_1)$ such that 
$e_1$ is a prefix of $e_2$, $f(e_1)$ is a prefix of $f(e_2)$.
\end{enumerate}

Strong linearizability has been shown to be equivalent to the existence 
of a forward simulation (defined below) from $O$ to an \emph{atomic object} 
$O(\mathit{Seq})$ defined by the set of sequential histories,
$\mathit{Seq}$~\cite{AttiyaE2019,Rady}.
Intuitively, if we consider an implementation of a sequential object with
histories in $\mathit{Seq}$, 
then the atomic object $O(\mathit{Seq})$ corresponds to running 
the same implementation in a concurrent context 
provided that method bodies execute in isolation.
Formally, the atomic object $O(\mathit{Seq})$ can be defined as a labeled 
transition system
where:
\begin{itemize}
	\item the set of states contains pairs formed of a history $h$ and 
    a linearization $h_s\in Seq$ of $h$,
	and the initial state contains an empty history and
    empty linearization, 
	\item the transition labels are call or return actions, or \emph{linearization point} actions $\mathit{lin}(k)$ for linearizing an invocation with identifier $k$, and
	\item the transition relation $\delta$ contains all the tuples $((h,h_s),a,(h',h_s'))$, where $a$ is a transition label, such that
\begin{align*}
& a\mbox{ is a call action} \implies h'=h\cdot a \mbox{ and } h_s'=h_s   \\
& a\mbox{ is a return action}\implies h'=h\cdot a\mbox{ and $h_s'=h_s$ and $a$ occurs in $h_s'$} \\
& a=lin(k) \implies h'=h \mbox{ and $h_s'=h_s\cdot \call{M(x)}_k\cdot \ret{y}_k$, for some $M$, $x$, and $y$.}
\end{align*}
Call actions are only appended to the history $h$, return actions ensure that the linearization $h_s'$ contains the corresponding method, and linearization  point actions extend the linearization with a new method.
\end{itemize}
The executions of $O(\mathit{Seq})$ are defined as sequences of transitions $s_0, a_0,s_1\ldots a_{k-1},s_k$, for some $k>0$, such that $(s_i, a_i, s_{i+1})\in \delta$ for each $0\leq i<k$. Note that $O(\mathit{Seq})$ admits every history which is linearizable w.r.t.~$\mathit{Seq}$, i.e., $H(O(\mathit{Seq}))=\{h: \exists h'\in\mathit{Seq}.\ h\sqsubseteq h'\}$.

Given two objects $O_1$ and $O_2$, a \emph{forward simulation} 
from $O_1$ to $O_2$ is a (binary) relation $F$ between states 
of $O_1$ and $O_2$ that maps
every step of $O_1$ to a possibly stuttering (no-op) step of $O_2$. 
Formally, $F$ is a forward simulation if it contains the pair of initial states of $O_1$ and $O_2$, and for every transition $(s_1,a,s_1')$ of $O_1$ between two states $s_1$ and $s_1'$ with label $a$ and every state $s_2$ of $O_2$ such that $(s_1,s_2)\in F$, there exists a state $s_2'$ of $O_2$ such that either:
\begin{itemize}
\item $s_2=s_2'$ (stuttering step) and $a$ is not a call or return action, or
\item $(s_1',s_2')\in F$,  $(s_2,a',s_2')$ is a transition of $O_2$, and if $a$ is a call or return action, then $a=a'$.
\end{itemize}

A forward simulation $F$ maps every transition of $O_1$ starting in a state $s_1$ to a transition of $O_2$ which starts in a state $s_2$ associated by $F$ to $s_1$. This is different from a related notion of \emph{backward} simulation that maps every transition of $O_1$ ending in a state $s_1'$ to a transition of $O_2$ ending in a state $s_2'$ associated by the simulation to $s_1'$ (see~\cite{DBLP:journals/iandc/LynchV95} for more details).

We say that $O_1$ \emph{strongly refines}
 $O_2$ when there exists a forward simulation from $O_1$ to $O_2$.
In the context of objects, a generic notion of refinement would correspond to the set of histories of $O_1$ being included in the set of histories of $O_2$, which is implied by but not equivalent to the existence of a forward simulation~\cite{AttiyaE2019,DBLP:journals/iandc/LynchV95}.
 We may omit the adjective strong for simplicity.

\subsection{Message-Passing Implementations}\label{ssec:mp_impls}

In message-passing implementations, methods can be invoked on a
distinguished set of processes called \emph{clients}. Clients are also
responsible for returning values of method invocations. The
interaction between invocations on different clients may rely on a
disjoint set of processes called \emph{servers}.
In general, we
assume that the processes are asynchronous
and communicate by sending and receiving messages that can experience
arbitrary delay but are not lost, corrupted, or spuriously generated.
Communication is permitted between any pair of processes,
not just between clients and servers.
A node may run any combination of a client process and a server process.
Processes are subject to crash failures;
we assume the client process and the server process running on the same
node can fail independently, which only strengthens our model.

To simplify the exposition, we model message-passing implementations using labeled transition systems instead of actual code.
Each process is defined by a transition system with states in an unspecified set $\states$.
A \emph{message} is a triple $\tup{\mathit{src},\mathit{dst},v}$ where $\mathit{src}$ is the sending process, $\mathit{dst}$ is the process to which the message is addressed, and $v$ is the message payload. The set of messages is denoted by $\msgs$.
The transition function $\delta_j$ of a server process $j$ is defined as a partial function $\delta_j:\states\times 2^{\msgs} \rightharpoonup \states\times 2^{\msgs}$. For a given local state $s$ and set of messages $\mathit{Msgs}$ received by $j$, $\delta_j(s,\mathit{Msgs})=(s',\mathit{Msgs}')$ defines the next local state $s'$ and a set of message $\mathit{Msgs}'$ sent by $j$. It is possible that $\mathit{Msgs}$ or $\mathit{Msgs}'$ is empty.
The transition function of a client $i$ is defined as $\delta_i:\states\times (2^{\msgs}\cup\acts) \rightharpoonup \states\times 2^{\msgs}$ where $\acts$ is a set of call and return actions.
Unlike servers, clients are allowed to perform additional \emph{method call} steps or \emph{method return} steps that are determined by call and return actions in $\acts$.
To simplify the presentation, we assume that a client state records whether
an invocation is currently
pending and what is the last returned value. Therefore, for a given state $s$ of a client $i$, $\mathit{pending}_i(s)=\mathit{true}$ iff
an invocation is currently pending in state $s$
and 
$\mathit{retVal}_i(s)=y$ iff there exists a state $s'$ such that $\delta_i(s',\ret{y})=(s,\_)$.

An implementation $\mpimp{m,n}$ with $m$ client processes and $n$ server processes is defined by an initial local state $s_0$ that for simplicity, we use to initiate the computation of all processes, and a set $\{\delta_k: 0\leq k < m+n\}$ of transition functions, where $\delta_k$, $0\leq k< m$, describe client processes and $\delta_k$, $m\leq k< m+n$, describe server processes.

\begin{figure} [t]
\small
  \centering
  \begin{mathpar}
    \inferrule[call]{ i < m\quad s_i=g(i) \downarrow_1\quad \mathit{pending}_i (s_i) = \mathit{false}\quad  \delta_i(s_i,\call{M(x)})=(s_i',\mathit{Msgs})}{
      g
      \xrightarrow{\call{M(x)}}_i
      g[i\mapsto \tup{s_i',(g(i)\downarrow_2 \cup \mathit{Msgs})}]
    }

    \inferrule[return]{ i < m\quad s_i=g(i) \downarrow_1\quad  \delta_i(s_i,\ret{y})=(s_i',\mathit{Msgs})}{
      g
      \xrightarrow{\ret{y}}_i
      g[i\mapsto \tup{s_i',(g(i)\downarrow_2 \cup \mathit{Msgs})}]
    }

    \inferrule[internal]{s_j=g(j) \downarrow_1\quad \mathit{Msgs}\subseteq ( \bigcup_{0\leq k< m+n} g(k)\downarrow_2 )\downarrow_{dst=j} \quad  \delta_i(s_j,\mathit{Msgs})=(s_j',\mathit{Msgs}')}{
      g
      \xrightarrow{ }_j
      g[j\mapsto \tup{s_j',(g(j)\downarrow_2 \cup \mathit{Msgs}')}]
    }
  \end{mathpar}
  \caption{State transitions of message-passing implementations. We define transitions using a standard notation where the conditions above the line must hold so that the transition given below the line is valid. For a function $f:A\rightharpoonup B$, $f[a\mapsto b]$ denotes the function $f':A\rightharpoonup B$ defined by $f'(c) = f(c)$, for every $c\neq a$ in the domain of $f$, and $f'(a)=b$. Also, for a tuple $t$, $t\downarrow_i$ denotes its $i$-th component, and for a set of messages $\mathit{Msgs}$, $\mathit{Msgs}\downarrow_{dst=j}$ is the set of messages in $\mathit{Msgs}$ with destination $j$.}
  \label{fig:op:sem:mp}
\end{figure}

The executions of a message-passing implementation $\mpimp{m,n}$ are interleavings of ``local'' transitions of individual processes.
A \emph{global state} $g$ is a function mapping each process to a local state and a pool of messages that the process sent since the beginning of the execution, i.e., $g:[0..m+n-1]\rightarrow \states\times 2^{\msgs}$. The initial global state $g_0$ maps each process to its initial local state and an empty pool of messages.
A transition between two global states advances one process according to its transition function. Figure~\ref{fig:op:sem:mp} lists the set of rules defining the transitions of $\mpimp{m,n}$. \textsc{call} and \textsc{return} transition rules correspond to steps of a client due to invoking or returning from a method, and \textsc{internal} represents steps of a client or a server where it advances its state due to receiving some set of messages. The set of received messages is chosen \emph{non-deterministically} from the pools of messages sent by all the other processes. The non-deterministic choice models arbitrary message delay since it allows sent messages to be ignored in arbitrarily many steps. The messages sent during a step of a process $i$ are added to the pool of messages sent by $i$ and never removed.

An \emph{execution} is a sequence of transition steps $g_0\xrightarrow{} g_1\xrightarrow{} \ldots$ between global states.
We assume that every message is \emph{eventually delivered}, i.e., for any infinite execution $e$, a transition step where a process $i$ sends a message $\mathit{msg}$ to a process $i'$ can \emph{not} be followed by an infinite set of steps of process $i'$ where the set of received messages in each step excludes $\mathit{msg}$.

\begin{remark}
For simplicity, our semantics allows
a message to be
delivered multiple times. We assume that the effects of message duplication can be avoided by including process identifiers and sequence numbers in message payloads. This way a process can track the set of messages it already received from any other process.
\end{remark}

We define a notion of crash fault tolerance for message-passing
implementations that asks for system-wide progress provided that at
most $f$ servers crash. Therefore, an implementation $\mpimp{m,n}$ is
\emph{$f$-nonblocking} iff for every infinite execution
$e=g_0\xrightarrow{} \ldots\xrightarrow{} g_k\xrightarrow{} \ldots$
and $k>0$ such that some invocation is pending in $g_k$,
if at least one client and $n-f$ servers execute a step
infinitely often in $e$, then some invocation completes after $g_k$
(i.e., the sequence of transitions in $e$ after $g_k$ includes a
\textsc{return} transition). 

For $m$ clients and $n$ servers, ABD (as well as its multi-writer version)
is $f$-nonblocking as long as $f < n/2$, while $m$ can be anything.
In fact, ABD provides a stronger liveness property,
in that every invocation by a non-faulty client eventually completes.
Furthermore, ABD only needs client-server communication.
So the communication model is weaker that what the model assumes and
the output is stronger than what the model requires.

\subsection{Shared Memory Implementations}\label{ssec:sm_impls}

In shared-memory implementations, the code of each method defines a sequence of invocations to a set of \emph{base} objects. In our work, the base objects are 
standard single-writer (SW)
registers. Methods can be invoked in parallel at a number of processes that are asynchronous and crash-prone. We assume that read and write accesses to SW registers are instantaneous.

We omit a detailed formalization of the executions of such an implementation. The pseudo-code we will use to define such implementations can be translated in a straightforward manner to executions seen as sequences of transitions between global states that track values of (local or shared) SW registers and the control point of each process.

We say that a shared-memory implementation is \emph{nonblocking} if
for every infinite execution $e=g_0\xrightarrow{} \ldots\xrightarrow{}
g_k\xrightarrow{} \ldots$ and $k>0$
such that some invocation is pending in $g_k$,
some invocation completes after $g_k$.
The definition of nonblocking for
shared-memory implementations demands system-wide progress even if all
processes but one fail.

\section{Shared-Memory Refinements of Message-Passing Implementations}\label{sec:ref}

We show that every message-passing object implementation with $m$ clients 
and any number $n$ of servers
can be refined by a shared-memory implementation with $m$ processes such that:
(1) the implementation uses only single-writer registers, and
(2) it is nonblocking if the message-passing implementation is 
$(m-1)$-nonblocking. 
By reduction from~\cite[Corollary~3.7]{HelmiHW12},
which shows that there is no nonblocking implementation of several
objects, including multi-writer registers, from single-writer
registers,
the existence of this refinement implies the impossibility of
strongly-linearizable message-passing implementations of the same objects
no matter how small the fraction of failures is.
This reduction relies on the equivalence between strong
linearizability and strong refinement and the compositionality of the
latter (see Section~\ref{sec:impossib}).

The shared-memory implementation should guarantee system-wide progress
even if all processes, except one, fail.
In contrast, the message-passing implementation only needs to guarantee
system-wide progress when no more than $f$ server processes fail.
Since the total number of servers may be arbitrarily larger than $f$,
it is impossible to define a ``hard-wired'' shared-memory refinement
where each shared-memory process simulates a pre-assigned message-passing
client or server process.
Instead, we have each of the $m$ shared-memory processes simulate a
client in the message-passing implementation while also cooperating
with the other processes in order to simulate steps of all the server
processes.
This follows the ideas in the BG simulation~\cite{BorowskyG1993}.
Overall, the shared-memory implementation simulates only a subset of
the message-passing executions, thereby, it is a \emph{refinement}
of the latter. The set
of simulated executions is however ``complete'' in the sense that a
method invocation is always enabled on a process that finished
executing its last invocation.

The main idea of the refinement is to use a hypothetical message-passing
implementation of an object using $m$ clients and $n$ servers
as a ``subroutine'' to implement the object
in a system with $m$ processes using SW registers.
Each process $p$ in the shared-memory algorithm is associated with a client
in the message-passing algorithm, and $p$, and only $p$, simulates the
steps of that client.
Since any number of shared-memory processes may crash, and any number of
message-passing clients may crash, this one-to-one association works fine.
However, the same approach will not work for simulating the message-passing
servers with the shared-memory processes, since the message-passing
algorithm might tolerate the failure of only a very small fraction of
servers, while the shared-memory algorithm needs to tolerate the failure
of all but one of its processes.
Instead, all the shared-memory processes cooperate to simulate each
of the servers.
To this end, each shared-memory process executes a loop in which it
simulates a step of its associated client, and then, for each one of the
servers in round-robin order, it works on simulating a step of that server.
The challenge is synchronizing the attempts by different shared-memory
processes to simulate the same step by the same server, without relying
on consensus.
We use \emph{safe agreement} objects to overcome this difficulty,
a separate one for the $r$-th step of server $j$, as follows:
Each shared-memory process proposes a value, consisting of
its local state and a set of messages to send,
for the $r$-th step of server $j$, and repeatedly checks (in
successive iterations of the outer loop) if the
value has been resolved, before moving on to the next step of server $j$.
Because of the definition of safe agreement, the only way that
server $j$ can be stuck at step $r$ is if one of the simulating shared-memory
processes crashes.

The steps of the client and server processes are handled in essentially
the same way by a shared memory-memory process, the main difference being
that client processes need to react to method invocations and provide
responses.
The current state of, and set of messages sent by, each
message-passing process is stored in a SW register.
The shared-memory process reads the appropriate register,
uses the message-passing transition function
to determine the next state and set of messages to send,
and then writes this information into the appropriate register.

More details follow, after we specify safe agreement.

\subsection{Safe Agreement Object}

The key to the cooperative simulation of server processes is a large
set of \emph{safe agreement} objects, each of which is used to agree on
a \emph{single} step of a server process. 
Safe agreement is a weak form of consensus that separates
the proposal of a value and
the learning of the decision into two methods.
A safe agreement object supports two wait-free methods,
{\em propose}, with argument $v \in V$
and return value {\em done}, and
{\em resolve}, with no argument
and return value 
$v \in V \cup \{\bot\}$.
While the methods are both wait-free, {\em resolve} may
return a ``non-useful'' value $\bot$.
Each process using such an object starts with an invocation
of {\em propose}, and
continues with a (possibly infinite) sequence of {\em resolve}
invocations;
in our simulation, {\em resolve} is not invoked after it
returns a value $v \ne \bot$.

The behavior of a safe agreement object is affected by the possible
crash of processes during a method.
Therefore, its correctness is not defined using linearizability
w.r.t.\ a sequential specification.
Instead, we define such an object to be correct when its (concurrent)
histories satisfy the following properties:
\begin{itemize}
\item {\em Agreement:}  If two {\em resolve} methods both return
non-$\bot$ values, then the values are the same.

\item {\em Validity:}  The return action of a {\em resolve} method
that returns a value $v \ne \bot$ is preceded by
a call action $\call{\mathit{propose}(v)}$.

\item {\em Liveness:}
If a {\em resolve} is invoked when there is no pending
{\em propose} method,
then it can return only a non-$\bot$ value.
\end{itemize}

The liveness condition for safe agreement is weaker than that for
consensus, as $\bot$ can be returned by {\em resolve} as long as a {\em
propose} method is pending.  Thus it is possible to implement
a safe agreement object using SW registers.
We present such an algorithm in Appendix~\ref{sec:sa}, based on those
in~\cite{BorowskyG1993,ImbsR2009}.

\subsection{Details of the Shared-Memory Refinement}

\renewcommand{\algorithmicrequire}{\textbf{Method}}

\begin{algorithm}[t]
\caption{Method $M$ at process $p_i$, $0\leq i < m$. Initially, resolved[$j$] is true and $r[j]$ is 0, for all $m\leq j<m+n$.}
\label{algorithm:simulation_new}
\small{
\begin{algorithmic}[1]
\Require $M(x)$:
\State client[$i$] $\gets \textsc{actStep}(\text{client[$i$]},\call{M(x)})$ \label{sim:call}  \Comment {simulating the call}
\While{true} \label{sim:after_call}
   \If{$\exists y.\ \delta_i(\text{client[$i$].state},\ret{y})$ is defined} \label{sim:return_condition}
      \State old\_client[$i$] $\gets$ client[$i$] \label{sim:before_ret} \Comment{used only to simplify the simulation relation}
      \State client[$i$] $\gets \textsc{actStep}(\text{client[$i$]},\ret{y})$ \label{sim:middle_ret} \Comment {simulating the return}
      \State \Return $y$  \label{sim:ret} 
   \EndIf
   \State client[$i$] $\gets \textsc{internalStep}(i)$ \label{sim:client}\Comment {simulating a step of client $i$}
   \For{$j \gets m,\ldots,m+n-1$}  \label{sim:loop} \Comment{simulate at most one step from each server}
           \If{resolved[$j$]} \label{sim:condition1}
                       \Comment{move on to next step of server $j$}
              \State s $\gets$ {\sc internalStep}($j$) \label{sim:server_step}
                       \Comment{returns a new state and pool of sent messages}
              \State r[$j$] $\gets$ r[$j$] + 1 \label{sim:inc}
              \State resolved[$j$] $\gets$ false
              \State SA[$j$][r[$j$]].propose(s) \label{sim:propose}
           \Else \Comment{keep trying to resolve current step of server $j$}
              \State s $\gets$ SA[$j$][r[$j$]].resolve()
              \If{s $\ne \bot$}  \label{sim:condition2} 
                 \State resolved[$j$] $\gets$ true
                 \State server[$i$][$j$] $\gets \langle \text{s, r[$j$]} \rangle$
                           \Comment{write to shared SW register}
                           \label{sim:server}
              \EndIf
	   \EndIf
    \EndFor
\EndWhile
\end{algorithmic}}
\end{algorithm}

\renewcommand{\algorithmicrequire}{\textbf{Function}}

\begin{algorithm}[t]
\caption{Auxiliary functions {\sc actStep}, {\sc internalStep}, and {\sc collectMessages}. $\textsc{mostRecent}$ is a declarative macro used to simplify the code.} 
\label{algorithm:internalStep}
\small{
\begin{algorithmic}[1]
\Require {\sc actStep}$(\text{client[$i$]},a)$:
\State \Return $\tup{\delta_i(\text{client[$i$].state}, a)\downarrow_1,\text{client[$i$].msgs}\cup \delta_i(\text{client[$i$].state}, a)\downarrow_2}$
\end{algorithmic}}
\vspace{2mm}
\small{
\begin{algorithmic}[1]
\Require {\sc internalStep}$(j)$ at process $p_i$, $0\leq i < m$:
\State Msgs $\gets$ {\sc collectMessages}($j$) \label{sim:finished_collection}
\If{$j < m$}\Comment{this is a client process}
\State ($q$, Msgs') $\gets$ $\delta_j$(client[$j$].state, Msgs)
         \label{line:transition1}
\Comment{determine new state and sent messages}
\State {\bf return} $\langle q, \text{client[$j$].msgs $\cup$ Msgs'} \rangle$
\Else\Comment{this is a server process}
\State ($q$, Msgs') $\gets$ $\delta_j$(server[$i$][$j$].state, Msgs)
         \label{line:transition2}
\Comment{determine new state and sent messages}
\State {\bf return} $\langle q, \text{server[$i$][$j$].msgs $\cup$ Msgs'} \rangle$
\EndIf
\end{algorithmic}}
\vspace{2mm}
\small{
\begin{algorithmic}[1]
\Require {\sc collectMessages}($j$):
\State Msgs $\gets \bigcup_{0\leq k\leq m-1}\text{client[$k$].msgs}\downarrow_{dst=j}$
             \Comment{identify messages sent to $j$ by clients}
\For{$k \gets m,\ldots,m+n-1$}
     \label{line:for-start}
\Comment{identify messages sent to $j$ by servers}
  \For{$i' \gets 0,\ldots,m-1$} \Comment{read the content of server registers}
      \State lserver[$i'$][$k$] $\gets$ server[$i'$][$k$]
  \EndFor
  \State s $\gets$ $\textsc{mostRecent}(\text{lserver}[0..m-1][k])$  \Comment{identify the most recent step of server $k$}
   \State Msgs $\gets$ Msgs $\cup$ s.msgs $\downarrow_{dst=j}$
\EndFor
\State \Return Msgs
\end{algorithmic}}

\vspace{2mm}
$\textsc{mostRecent}(\text{lserver[0..$m$-1][$k$]})$ = (lserver[$i$][$k$].state, lserver[$i$][$k$].msgs) such that lserver[$i$][$k$].sn =
$max_{0\leq j\leq m-1} \text{lserver[$j$][$k$].sn}$
\end{algorithm}

Let $\mpimp{m,n}$ be a message-passing implementation. We define a shared-memory implementation $\smimp{m}$ that refines $\mpimp{m,n}$ and that runs over a set of processes $p_i$ with $0\leq i <m$. Each process $p_i$ is associated with a client $i$ of $\mpimp{m,n}$. The code of a method $M$ of $\smimp{m}$ executing on a process $p_i$ is listed in Algorithm~\ref{algorithm:simulation_new}. This code uses the following set of shared objects (the other registers used in the code are local to a process):
\begin{itemize}
\item client[$i$]: SW register written by $p_i$, holding the current local state (accessed using .state) and pool of sent messages (accessed using .msgs) of client $i$; $0\leq i< m$
\item server[$i$][$j$]: SW register written by $p_i$, holding the current state and pool of sent messages of server $j$
according to $p_i$, 
tagged with a step number (accessed using .sn); $0\leq i< m$ and $m\leq j< m+n$
\item SA[$j$][$r$]: safe agreement object used to agree on the $r$-th step of server $j$ ($m\leq j< m+n$ and $r = 0, 1, \ldots$).
\end{itemize}
Initially, client[$i$] stores the initial state and an empty set of messages, for every $0\leq i<m$. Also, server[$i$][$j$] stores the initial state,  an empty set of messages, and the step number 0, for every $0\leq i<m$ and $m\leq j< m+n$.

A process $p_i$ executing a method $M$ simulates the steps that client $i$ would have taken when the same method $M$ is invoked. It stores the current state and pool of sent messages in client[$i$]. Additionally, it contributes to the simulation of server steps. Each process $p_i$ computes a proposal for the $r$-th step of a server $j$ (the resulting state and pool of sent messages -- see line~\ref{sim:server_step}) and uses the safe agreement object SA[$j$][$r$] to reach agreement
with the other processes (see line~\ref{sim:propose}). It computes a proposal for a next step of server $j$ only when agreement on the $r$-th step has been reached, i.e., it gets a non-$\bot$ answer from SA[$j$][$r$].resolve() (see the if conditions at lines~\ref{sim:condition1} and~\ref{sim:condition2}). However, it can continue proposing or agreeing on steps of other servers. It iterates over all server processes in a round-robin fashion, going from one server to another when resolve() returns $\bot$. This is important to satisfy the desired
progress guarantees.

Steps of client or server processes are computed locally using the transition functions of $\mpimp{m,n}$ in \textsc{actStep} and \textsc{internalStep}, listed in Algorithm~\ref{algorithm:internalStep}. A method $M$ on a process $p_i$ starts by advancing the state of client $i$ by simulating a transition labeled by a call action (line~\ref{sim:call}). To simulate an ``internal'' step of client $i$ (or a server step), a subtle point is computing the set of messages that are supposed to be received in this step. This is done by reading all the registers client[\_] and server[\_][\_] in a sequence 
and collecting the set of messages in client[\_].msgs or server[\_][\_].msgs that have $i$ as a destination. Since the shared-memory processes can be arbitrarily slow or fast in proposing or observing agreement on the steps of a server $j$, messages are collected only from the ``fastest'' process, i.e., the process $p_k$ such that server[$k$][$j$] contains the largest step number among server[0..$m$-1][$j$] (see the \textsc{mostRecent} macro). This is important to ensure that messages are eventually delivered. Since the set of received messages contains \emph{all} the messages from client[\_].msgs or server[\_][\_].msgs with destination $i$ as opposed to a non-deterministically chosen subset (as in the semantics of $\mpimp{m,n}$ -- see Figure~\ref{fig:op:sem:mp}), some steps of $\mpimp{m,n}$ may not get simulated by this shared-memory implementation. However, this is not required as long as the shared-memory implementation allows methods to be invoked arbitrarily on ``idle'' processes (that are not in the middle of another invocation). This is guaranteed by the fact that each client is simulated locally by a different shared-memory process.
A process $p_i$ returns whenever a return action is enabled in the current state stored in client[$i$] (see the condition at line~\ref{sim:return_condition}). Server steps are computed in a similar manner to ``internal'' steps of a client.

\subsection{Correctness of the Shared-Memory Refinement}

We prove that there exists a forward simulation from the shared-memory implementation defined in Algorithm~\ref{algorithm:simulation_new} to the underlying message-passing implementation $\mpimp{m,n}$, which proves that the former is a (strong) \emph{refinement} of the latter. The proof shows that roughly, the message passing state defined by the content of all registers client[$i$] with $0\leq i<m$ and the content of all registers server[$i$][$j$] that have the highest step number among server[$i'$][$j$] with $0\leq i'< m$ is reachable in $\mpimp{m,n}$.
Each write to a register client[$i$] corresponds to a transition in the message-passing implementation that advances the state of client $i$, and each write to server[$i$][$j$]
containing a step number that is written for the first time among all writes to server[\_][j] corresponds to a transition that advances the state of server $j$.
This choice is justified since the same value is written in these
writes, by properties of safe agreement.
Then, we also prove that $\smimp{m}$ is nonblocking provided that $\mpimp{m,n}$ is $(m-1)$-nonblocking.

\begin{theorem}\label{th:sim:safety}
$\smimp{m}$ is a refinement of $\mpimp{m,n}$.
\end{theorem}
\begin{proof}
We define a relation $F$ between shared-memory and message-passing global states as follows: every global state of Algorithm~\ref{algorithm:simulation_new} is associated by $F$ with a message-passing global state $g$ such that for every client process $0\leq i< m-1$ and server process $m\leq j < n$,
$$
g(i) = \left\{\begin{array}{ll} \textsc{actStep}(\text{client[$i$]},\call{M(x)}),& \mbox{ if $p_i$ is before control point~\ref{sim:after_call} in Algorithm~\ref{algorithm:simulation_new}} \\
\text{old\_client[$i$]},& \mbox{ if $p_i$ is at control points~\ref{sim:middle_ret} or \ref{sim:ret} in Algorithm~\ref{algorithm:simulation_new}}\\
\text{client[$i$]},& \mbox{ otherwise}\end{array}\right.
$$
$$
g(j) = \textsc{mostRecent}(\text{server}[0..m-1][j])
$$

The first two cases in the definition of $g(i)$ are required so that call and return transitions in shared-memory are correctly mapped to call and return transitions in message-passing. The first case concerns call transitions and intuitively, it provides the illusion that a shared-memory call and the first statement in the method body (at line~\ref{sim:call}) are executed instantaneously at the same time. The second case concerns return transitions and ``delays'' the last statement before return (at line~\ref{sim:middle_ret}) so that it is executed instantaneously with the return.

Note that $F$ is actually a function since the message-passing global state is uniquely determined by the process control points and the values of the registers in the shared-memory global state. Also, the use of \textsc{mostRecent} is well defined because server[$i$][$j$].sn $=$ server[$i'$][$j$].sn implies that server[$i$][$j$].state $=$ server[$i'$][$j$].state and server[$i$][$j$].msgs $=$ server[$i'$][$j$].msgs, for every $0\leq i, i'< m$ (due to the use of the safe agreement objects).

In the following, we show that $F$ is indeed a forward simulation. Let us consider an indivisible step of Algorithm~\ref{algorithm:simulation_new} going from a global state $v_1$ to a global state $v_2$, and $g_1$ the message-passing global state associated with $v_1$ by $F$. We show that going from $g_1$ to the message-passing global state $g_2$ associated with $v_2$ by $F$ is a valid (possibly stuttering) step of the message-passing implementation. We also show that call and return steps of Algorithm~\ref{algorithm:simulation_new} are simulated by call and return steps of the message-passing implementation, respectively.

We start the proof with call and return steps. Thus, consider a step of Algorithm~\ref{algorithm:simulation_new} going from $v_1$ to $v_2$ by invoking a method $M$ with argument $x$ on a process $p_i$. Invoking a method in Algorithm~\ref{algorithm:simulation_new} will only modify the control point of $p_i$. Therefore, the message-passing global states $g_1$ and $g_2$ differ only with respect to process $i$: $g_1(i)$ is the value of client[$i$] in $v_1$ while $g_2(i)$ is the result of $\textsc{actStep}$ on that value and $\call{M(x)}$ (since the process is before control point~\ref{sim:after_call}). Therefore, $g_1\xrightarrow{\call{M(x)}}_i g_2$ (cf. Figure~\ref{fig:op:sem:mp}). For return steps of Algorithm~\ref{algorithm:simulation_new}, $g_1$ and $g_2$ also differ only with respect to process $i$: $g_1(i)$ is the value of old\_client[$i$] in $v_1$ (since the process is at control point~\ref{sim:ret}) while $g_2(i)$ is the value of client[$i$] in $v_2$. From lines~\ref{sim:before_ret}--\ref{sim:ret} of Algorithm~\ref{algorithm:simulation_new}, we get that the value of client[$i$] in $v_2$ equals the value of $\textsc{actStep}$ for $\text{old\_client[$i$]}$ in $v_1$ and the action $\ret{y}$ (note that old\_client[$i$] and client[$i$] are updated only by the process $p_i$). Therefore, $g_1\xrightarrow{\ret{y}}_i g_2$ (cf. Figure~\ref{fig:op:sem:mp}).

Every step of Algorithm~\ref{algorithm:simulation_new} except for the writes to client[$i$] or server[$i$][$j$] at lines~\ref{sim:client} and~\ref{sim:server} is mapped to a stuttering step of the message-passing implementation. This holds because $F$ associates the same message-passing global state to the shared-memory global states before and after such a step.

Let us consider a step of Algorithm~\ref{algorithm:simulation_new} executing the write to client[$i$] at line~\ref{sim:client} (we refer to the write that happens once $\textsc{internalStep}(i)$ has finished -- we do \emph{not} assume that line~\ref{sim:client} happens instantaneously). We show that it is simulated by a step of client $i$ of the message-passing implementation. By the definition of {\sc internalStep}, the value of client[$i$] in $v_2$ is obtained by applying the transition function of process $i$ on the state stored in client[$i$] of $v_1$ and some set of messages $\mathit{Msgs}$ collected from client[$i'$] and server[$i'$][$j$] with $0\leq i'<m$ and $m\leq j<n$. $\mathit{Msgs}$ is computed using the function {\sc collectMessages} that reads values
of client[$i'$] and server[$i'$][$j$] in shared-memory states that may precede $v_1$. However, since the set of messages stored in each of these registers increases monotonically\footnote{This is a straightforward inductive invariant of Algorithm~\ref{algorithm:simulation_new}.}, $\mathit{Msgs}$ is included in the set of messages stored in $v_1$ (i.e., the union of client[$i'$].msgs and server[$i'$][$j$].msgs for all $0\leq i'<m$ and $m\leq j<n$). Therefore,
\begin{align*}
\mathit{Msgs}\subseteq ( \bigcup_{0\leq k< n} g_1(k)\downarrow_2 )\downarrow_{dst=j},
\end{align*}
which together with the straightforward application of $\delta_i$ in {\sc internalStep} implies that $g_1\xrightarrow{ }_i g_2$.

Finally, let us consider a step of Algorithm~\ref{algorithm:simulation_new} executing the write to server[$i$][$j$] at line~\ref{sim:server}. Let $\tup{s,t}$ be the value written to server[$i$][$j$] in this step.
If there exists some other process $p_{i'}$ such that the register server[$i'$][$j$] in $v_1$ stores a tuple $\tup{s',t'}$ with $t\leq t'$, then this step is mapped to a stuttering step of the message-passing implementation. Indeed, the use of \textsc{mostRecent} in the definition of $F$ implies that it associates the same message-passing global state to the shared-memory states before and after such a step.
Otherwise, we show that this write is simulated by a step of server $j$ of the message-passing implementation. By the specification of the safe agreement objects, $s$ is a proposed value, and therefore, computed using {\sc internalStep} by a possibly different process $p_{i'}$. During this {\sc internalStep} computation server[$i'$][$j$] stores a value of the form $\tup{s',t-1}$, for some $s'$ (cf. the increment at line~\ref{sim:inc}). Since the values stored in the server[$i'$][$j$] registers are monotonic w.r.t.~their step number component, it must be the case that $s'$ is the outcome of $\textsc{mostRecent}(\text{server}[0..m-1][j])$ when applied on the global state $v_1$. Therefore, the {\sc internalStep} computation of $p_{i'}$ applies $\delta_j$ on the state $g_1(j)\downarrow_1$ and a set of messages $\mathit{Msgs}$ computed using {\sc collectMessages}. As in the case of the client[$i$] writes,
\begin{align*}
\mathit{Msgs}\subseteq ( \bigcup_{0\leq k< n} g_1(k)\downarrow_2 )\downarrow_{dst=j},
\end{align*}
which implies that $g_1\xrightarrow{ }_j g_2$.
\end{proof}

The message-passing executions simulated by the shared-memory executions
satisfy the eventual message delivery assumption. Indeed, since all the shared objects are wait-free, a message $\mathit{msg}$ stored in client[$i$] or server[$i$][$j$] will be read by all non-failed processes in a finite number of steps. Therefore, if $\mathit{msg}$ is sent to a client process $i'$, then it will occur in the output of $\textsc{internalStep}(i')$ at line~\ref{sim:client} on process $p_{i'}$ after a finite number of invocations of this function. Also, if $\mathit{msg}$ is sent to a server process $j'$, then it will be contained in the output of $\textsc{internalStep}(j')$ at line~\ref{sim:server_step} on every non-failed process $p_{i'}$ with $0\leq i'<m$ after a finite number of steps.

In the following, we show that the shared-memory implementation
is nonblocking (guarantees system-wide progress for $m$ processes, any
number of which can fail)
assuming that the message-passing implementation guarantees system-wide
progress if at most $m-1$ servers fail.

\begin{theorem}\label{th:sim:liveness}
If $\mpimp{m,n}$ is $(m-1)$-nonblocking, then $\smimp{m}$ is nonblocking. 
\end{theorem}
\begin{proof}
Since Algorithm~\ref{algorithm:simulation_new} uses only wait-free objects (SW registers and safe agreement objects), an invocation of a method $M$ at a non-crashed process could be non-terminating only because the \emph{resolve} invocations on safe agreement objects return $\bot$ indefinitely.
The latter could forbid the progress of a single server process. By the specification of safe agreement, \emph{resolve} can return $\bot$ only if it started while a \emph{propose} invocation (on the same object) is pending. Since a process $p_i$ has at most one invocation of \emph{propose} pending at a time, the number of \emph{propose} invocations that remain unfinished indefinitely is bounded by the number of failed shared-memory processes. Therefore, $m-1$ failed shared-memory processes forbid progress on at most $m-1$ server processes. Since, $\mpimp{m,n}$ is $(m-1)$-nonblocking, we get that $\smimp{m}$ is nonblocking.
\end{proof}

The proof above also applies to an extension of Theorem~\ref{th:sim:liveness} to wait-freedom, i.e., $\smimp{m}$ is wait-free if $\mpimp{m,n}$ ensures progress of individual clients assuming at most $m-1$ server failures.

\section{Impossibility Results}\label{sec:impossib}

We show the impossibility of strongly-linearizable nonblocking implementations in an asynchronous message-passing system for several highly useful objects (including multi-writer registers). This impossibility result is essentially a reduction from~\cite[Corollary~3.7]{HelmiHW12} that states a corresponding result for shared-memory systems.
Since strong linearizability and (strong) refinement are equivalent and refinement is compositional~\cite{AttiyaE2019,DBLP:journals/iandc/LynchV95,Rady},
the results in Section~\ref{sec:ref} imply that any strongly-linearizable message-passing implementation can be used to define a strongly-linearizable implementation in shared-memory. Since the latter also preserves the nonblocking property, the existence of a message-passing implementation would contradict the shared-memory impossibility result.

\begin{theorem}\label{th:imposs:reduction}
Given a sequential specification $\mathit{Seq}$, there is a shared-memory implementation with $m$ processes, which is nonblocking, strongly linearizable w.r.t.~$\mathit{Seq}$, and which only uses SW registers, if there is a message-passing implementation with $m$ clients and an arbitrary number $n$ of servers, which is $(m-1)$-nonblocking and strongly linearizable w.r.t.~$\mathit{Seq}$.
\end{theorem}
\begin{proof}
Given a message-passing implementation $\mpimp{m,n}$ as above, Theorem~\ref{th:sim:safety} and Theorem~\ref{th:sim:liveness} show that the shared-memory implementation $\smimp{m}$ defined in Algorithm~\ref{algorithm:simulation_new} is a refinement of $\mpimp{m,n}$ and nonblocking. Since strong linearizability w.r.t.~$\mathit{Seq}$ is equivalent to refining $O(\mathit{Seq})$ (see Section~\ref{sec:objects}) and the refinement relation (defined by forward simulations) is transitive\footnote{If there is a forward simulation $F_1$ from $O_1$ to $O_2$ and a forward simulation $F_2$ from $O_2$ to $O_3$, then the composition $F_1\circ F_2 = \{(s_1,s_3): \exists s_2. (s_1,s_2)\in F_1\land (s_2,s_3)\in F_2\}$ is a forward simulation from $O_1$ to $O_3$.}, we get that $\smimp{m,n}$ is a refinement of $O(\mathit{Seq})$, which implies that it is strongly linearizable w.r.t.~$\mathit{Seq}$. Finally, Theorem~\ref{th:sa:existence} shows that the safe agreement objects in $\smimp{m}$ can be implemented only using SW registers, which implies that $\smimp{m}$ only relies on SW registers.
\end{proof}

\begin{corollary}
There is no strongly linearizable message-passing implementation with 
$m\geq 3$ clients of multi-writer registers, max-registers, counters, or snapshot objects, which is $(m-1)$-nonblocking.
\end{corollary}
\begin{proof}
If such an implementation existed, then Theorem~\ref{th:imposs:reduction} would imply the existence of a strongly linearizable nonblocking implementation from single-writer registers, which is impossible by~\cite[Corollary~3.7]{HelmiHW12}.
\end{proof}

\section{Conclusions and Related Work}
\label{section:related}

In order to exploit composition and abstraction in
message-passing systems, it is crucial to understand how properties of
randomized programs are preserved when they are composed with object
implementations.
This paper extends the study of strong linearizability to message-passing
object implementations, showing how results for shared-memory object
implementations can be translated.
Consequently, there can be no strongly-linearizable crash-tolerant
message-passing implementations of multi-writer registers, max-registers,
counters, or snapshot objects.

In the context of shared-memory object implementations,
several results have shown the limitations of
strongly-linearizable implementations.
Nontrivial objects,
including multi-writer registers, max registers, snapshots, and counters,
have no nonblocking strongly-linearizable implementations
from single-writer registers~\cite{HelmiHW12}.
In fact, even with multi-writer registers,
there is no wait-free strongly-linearizable implementation
of a monotonic counter~\cite{DenysyukW15},
and, by reduction, neither of snapshots nor of max-registers.
Queues and stacks do not have an $n$-process nonblocking
strongly-linearizable implementation from objects whose readable
versions have consensus number less than $n$~\cite{AttiyaCH2018}.

On the positive side,
any consensus object is strongly linearizable,
which gives an \emph{obstruction-free} strongly-linearizable
universal implementation (of any object)
from single-writer registers~\cite{HelmiHW12}.
Helmi et al.~\cite{HelmiHW12} also give a \emph{wait-free}
strongly-linearizable implementation of \emph{bounded} max register
from multi-writer registers~\cite{HelmiHW12}.
When updates are strongly linearizable,
objects have \emph{nonblocking} strongly-linearizable implementations
from multi-writer registers~\cite{DenysyukW15}.
The space requirements of the latter implementation is avoided in
a \emph{nonblocking} strongly-linearizable implementation of
snapshots~\cite{OvensW19}.
This snapshot implementation is then employed with an algorithm
of~\cite{AspnesH1990wait}
to get a \emph{nonblocking} strongly-linearizable universal implementation
of any object in which all methods either commute or overwrite.

The \emph{BG simulation} has been used in many situations
and several communication models.
Originally introduced for the shared-memory model~\cite{BorowskyG1993},
it showed that $t$-fault-tolerant algorithms to solve \emph{colorless}
tasks (like set agreement) among $n$ processes,
can be translated into $t$-fault-tolerant algorithms for $t+1$ processes
(i.e., wait-free algorithms) for the same problem.
The \emph{extended} BG simulation~\cite{Gafni2009}
also works for so-called \emph{colored} tasks,
where different processes must decide on different values.
Another extension of the BG simulation~\cite{FraigniaudGRR2014}
was used to dynamically reduce synchrony of a system.
(See additional exposition in~\cite{ImbsR2009,Kuznetsov2013}.)

To the best of our knowledge,
all these simulations allow only a single invocation by each process,
and none of them handles \emph{long-lived} objects.
Furthermore, they are either
among different variants of the shared-memory model~\cite{BorowskyG1993,FraigniaudGRR2014,Gafni2009,Kuznetsov2013}
or among different failure modes in the message-passing model~\cite{DolevG2016,ImbsRS2016}.

This paper deals with multi-writer registers
and leaves open the question of finding a strongly-linearizable
message-passing implementation of a \emph{single-writer} register.
The original ABD register implementation~\cite{AttiyaBD1995},
which is for a single writer,
is not strongly linearizable~\cite{HadzilacosHT2020arxiv4}.
Recent work~\cite{ChanHHT2021} addresses this open question, showing
that a simple object called ``test or set'' has no strongly-linearizable
message-passing implementation and then appealing to the fact that
a single-writer register is a generalization of this object.

Recently, two ways of mitigating the bad news of this paper have been
proposed, both of which move away from strong linearizability.
In~\cite{HadzilacosHT2021}, a consistency condition that is
intermediate between linearizability and strong linearizability, called
``write strong-linearizability'' is defined and it is shown that
for some program this condition is
sufficient to preserve the property of having non-zero termination probability,
and that a variant of ABD satisfies write strong-linearizability.
In another direction, \cite{AttiyaEW2021blunting} presents a simple
modification to ABD that preserves the
property of having non-zero termination probability; the modification is
to query the servers multiple times instead of just once and then randomly
pick which set of responses to use.
This modification also applies to the snapshot implementation 
in~\cite{AfekADGMS1993}; note that snapshots 
do not have nonblocking strongly-linearizable implementations,
in either shared-memory (proved in~\cite{HelmiHW12})
or message-passing (as we prove in this paper, by reduction).

\bibliographystyle{plainurl}
\bibliography{citations}

\appendix
\section{An Implementation of Safe Agreement}\label{sec:sa}

We present an algorithm to implement a safe agreement
object that only uses \emph{single-writer} registers.
The algorithm is based on \cite{BorowskyG1993,ImbsR2009}.

The crux of the safe agreement algorithm is to identify a
\emph{core} set of processes,
roughly, those who were first to start the algorithm.
Once the core set is identified, the proposal of a fixed process
in this set is returned.
Our algorithm picks the proposal of the process with minimal id,
but the process with maximal id can be used just as well.
A ``double collect'' mechanism is used to identify the core set,
by having every process
write its id and repeatedly read all the processes' corresponding
variables until it observes no change.\footnote{
    This use of double collect is a stripped-down version of
    the snapshot algorithm~\cite{AfekADGMS1993}.}
The process then writes the set consisting of all the ids collected.
To resolve, a process reads all these sets, and intuitively,
wishes to take the smallest set among them, $C$, as the core set.
However, it is possible that an even smaller set will be written later.
The key insight of the algorithm (identified by~\cite{BorowskyG1993})
is that such a smaller set can only be written by a process whose
identifier is already in $C$.
Thus, once all processes in $C$ wrote their sets,
either one of them is strictly contained in $C$
(and hence, can replace it),
or no smaller set will ever be written.

The pseudocode is listed in Algorithm~\ref{algorithm:SA}.
The algorithm uses the following single-writer shared registers
(the other registers used in the code are local to a process):
\begin{itemize}
\item {\em Val}$[i]$: register written by $p_i$, holding a proposal,
   initially $\bot$; $0 \le i < m$
\item {\em Id}$[i]$: register written by $p_i$, holding its own id; initially $\bot$;
   $0 \le i < m$
\item {\em Set}$[i]$: register written by $p_i$, holding a set of process ids,
   initially $\emptyset$; $0 \le i < m$
\end{itemize}

\begin{algorithm}[tb]
\caption{Safe agreement, code for process $p_i$.}
\label{algorithm:SA}
\begin{algorithmic}[1]

\State {\bf Propose}($v$):
\State {\em Val}$[i] \gets v$ \Comment{ announce own proposal }
\State {\em Id}$[i] \gets i$ \Comment{ announce own participation }
        \label{line:sa-write-id}
\Repeat \Comment{ double collect }
        \label{line:prop-repeat-start}
\State {\bf for} $j \gets 0,\ldots,m-1$ {\bf do}
               {\em collect}$1[j] \gets$ {\em Id}$[j]$
\State {\bf for} $j \gets 0,\ldots,m-1$ {\bf do}
               {\em collect}$2[j] \gets$ {\em Id}$[j]$
\Until {{\em collect}$1 =$ {\em collect}$2$}
                \Comment{ all components are equal }
        \label{line:prop-repeat-end}
\State {\em Set}$[i] \gets \{ j :$ {\em collect}$1[j] \neq \bot \}$
        \label{line:sa-write-set}
\Statex
\State {\bf Resolve}():
\State {\bf for} $j \gets 0,\ldots,m-1$ {\bf do}
               $s[j] \gets$ {\em Set}$[j]$ \Comment{ read $m$ registers }
       \label{line:sa-read-sets}
\State $C \gets$ smallest (by containment) non-empty set in
          $s[0,\ldots,m-1]$
\If{for every $j \in C$, $( ( s[j] \neq \emptyset )$ and
                         $( C \subseteq s[j] ) )$} \label{line:sa-return-cond}
        \State \Return {\em Val}$[\min(C)]$
               \label{line:sa-return-val}
        \Comment{ the proposal of the process with minimal id in $C$ }
\Else \State \Return $\bot$ \Comment{ no decision yet }
        \label{line:sa-return-bot}
\EndIf
\end{algorithmic}
\end{algorithm}

Notice that {\em propose} and {\em resolve} are wait-free.
This is immediate for {\em resolve}.
For {\em propose}, note that the double collect loop
(in Lines~\ref{line:prop-repeat-start}--\ref{line:prop-repeat-end})
is executed at most $m$ times,
since there are at most $m$ writes to {\em Id}
(one by each process).

\begin{theorem}
\label{th:sa:existence}
Algorithm~\ref{algorithm:SA} is an implementation of safe agreement
from single-writer registers.
\end{theorem}

\begin{proof}
To show validity,
first note that a non-$\bot$ value $v$ returned by any {\em resolve}
method is that stored in {\em Val}$[i]$ for some $i$ such that
$p_i$ wrote to its {\em Id}$[i]$ shared variable.
The code ensures that before $p_i$ writes to {\em Id}$[i]$,
it has already written $v$ to {\em Val}$[i]$,
in response to the invocation of {\em propose}$(v)$.

Agreement and liveness hinge on the following comparability
property of the sets of ids written to the array {\em Set}:

\begin{lemma}
\label{lemma:comparable sets}
For any two processes $p_i$ and $p_j$, if
$p_i$ writes $S_i$ to {\em Set}$[i]$ and
$p_j$ writes $S_j$ to {\em Set}$[j]$,
then either $S_i \subseteq S_j$ or $S_j \subseteq S_i$.
\end{lemma}

\begin{proof}
Assume by contradiction that $S_i$ and $S_j$ are incomparable, i.e., there exist $i'\in S_i\setminus S_j$ and $j'\in S_j\setminus S_i$.
Without loss of generality, let us assume that $p_{i'}$ writes its id to {\em Id}$[i']$ before $p_{j'}$ writes its id to {\em Id}$[j']$ (otherwise, a symmetric argument applies).
Since $j'\in S_j$, the last collect in the loop at line~\ref{line:prop-repeat-start} on process $p_j$ starts after $p_{j'}$ writes to {\em Id}$[j']$ and $p_{i'}$ writes to {\em Id}$[i']$. Therefore, the process $p_j$ must have read $i'$ from {\em Id}$[i']$ in this collect (i.e., {\em collect}$2[i']$ = $i'$), which contradicts the assumption that $i'\not\in S_j$.
\end{proof}

Suppose $p_i$ returns a non-$\bot$ value {\em Val}$[k]$
because $k$ is the smallest id in $C = s[h]$,
which is the smallest Set read by $p_i$ in Line~\ref{line:sa-read-sets},
and $p_{i'}$ returns a non-$\bot$ value {\em Val}$[k']$
because $k'$ is the smallest id in $C' = s[h']$,
which is the smallest Set read by $p_{i'}$ in Line~\ref{line:sa-read-sets}.
Assume in contradiction that $k \ne k'$,
which implies that $C \ne C'$
and $h \ne h'$.
By Lemma~\ref{lemma:comparable sets}, $C$ and $C'$ are comparable;
without loss of generality, assume $C \subseteq C'$.
Then $C\subset C'$, which contradicts the condition for returning a non-$\bot$ value (Line~\ref{line:sa-return-cond}) in $p_{i'}$. Indeed, since $h\in C\subset C'$ (every process reads \emph{Id} registers after writing to its own), $p_{i'}$ should have read {\em Set}$[h]$ before returning and witnessed the fact that it contains a smaller set than {\em Set}$[h']$.

We now consider liveness.
Assume no process has an unfinished {\em propose} method.
Thus, every process that writes to its {\em Id} variable in
Line~\ref{line:sa-write-id},
also writes to its {\em Set} variable in Line~\ref{line:sa-write-set}.
Consider any {\em resolve} method, say by $p_i$,
that begins after the last {\em propose} method completes.
Let $C$ be the smallest non-empty set obtained by $p_i$ in
Line~\ref{line:sa-read-sets}.
For each $j \in C$, {\em Set}$[j]$ is not empty,
since all the {\em propose} methods completed.
By the choice of $C$, Lemma~\ref{lemma:comparable sets} ensures that $C$ is
a subset of {\em Set}$[j]$.
Thus $p_i$ returns a non-$\bot$ value in Line~\ref{line:sa-return-val}.
\end{proof}

\end{document}